\documentclass[final]{dmtcs-episciences}

\usepackage{style}

\title[On the End-Vertex Problem of Graph Searches]{On the End-Vertex Problem of Graph Searches
}
	
	\author[J.~Beisegel, C.~Denkert, E.~K\"ohler, M.~Krnc, N.~Piva\v c, R.~Scheffler, M.~Strehler]{Jesse Beisegel\affiliationmark{1} \and Carolin Denkert\affiliationmark{1} \and Ekkehard K\"ohler\affiliationmark{1} \and Matja\v{z} Krnc\affiliationmark{2,3}
   \\\and Nevena Piva\v c\affiliationmark{2}
    \and Robert Scheffler\affiliationmark{1}
  \and Martin Strehler\affiliationmark{1}
  }
	
  \affiliation{
	Brandenburg University of Technology, Cottbus, Germany \\
	University of Primorska, Koper, Slovenia \\
  	Faculty of Information Studies, Novo mesto, Slovenia
  }

  \received{2018-10-31}
  \revised{2019-5-21}
  \accepted{2019-5-29}

\begin{document}

	\publicationdetails{21}{2019}{1}{13}{4937}
  \maketitle

\begin{abstract}
	End vertices of graph searches can exhibit strong structural properties and are crucial for many graph algorithms. The problem of deciding whether a given vertex of a graph is an end-vertex of a particular search was first introduced by Corneil, K\"ohler, and Lanlignel in 2010. There they showed that this problem is in fact \NP-complete for LBFS on weakly chordal graphs. A similar result for BFS was obtained by Charbit, Habib and Mamcarz in 2014. Here, we prove that the end-vertex problem is \NP-complete for MNS on weakly chordal graphs and for MCS on general graphs. Moreover, building on previous results, we show that this problem is linear for various searches on split and unit interval graphs.
\end{abstract} 	
  
\section{Introduction}

Graph search is one of the oldest and most fundamental algorithmic concepts  in both
graph theory and computer science. In essence, it is the systematic exploration of the
vertices in a graph, beginning at a chosen vertex and visiting every other vertex
of the graph such that a vertex is visited only if it is adjacent to a previously visited
vertex; this procedure is known as \emph{Generic Search}.

Such a general definition of a graph search leaves much freedom for a selection rule determining which vertex is chosen next. 
By restricting this choice with specific rules, various variants of graph
searches can be defined. Arguably the earliest examples are the \emph{Breadth First Search} (BFS) 
and the \emph{Depth First Search} (DFS), which were referred to in the context of maze traversals in the 19th century~\cite{even2011book}. 
Among the more sophisticated searches are, for example, \emph{Lexicographic Breadth First Search} (LBFS)~\cite{rose1976}  and 
\emph{Lexicographic Depth First Search} (LDFS)~\cite{corneil2008unified}. In this paper, we will study \emph{Maximum Cardinality Search} (MCS)~\cite{tarjan1984simple} and \emph{Maximal Neighborhood Search} (MNS)~\cite{corneil2008unified} in more detail.

Although being rather simple algorithms, the power of graph searches should not be underestimated. Using such simple procedures, a variety of important graph properties can be tested, many optimization problems can be simplified, and they are sub-routines in many more algorithms. 
Among such problems are topological ordering, finding connected components,
testing bipartiteness, computing the shortest paths with respect to the number of edges, 
and the Edmonds-Karp algorithm for computing the maximum flow in a network~\cite{edmonds1972theoretical}.

Usually, the outcome of a graph search is a \emph{search order}, i.e., a sequence of the 
vertices in order of visits. There are many results using such orders. For instance, by reversing an 
LBFS order of a chordal graph, one finds a perfect elimination order of this graph~\cite{rose1976}. 
A perfect elimination order in $G=(V,E)$ is an ordering of the vertices such that 
for every vertex $v \in V$ the neighbors of $v$ that occur after $v$ in the order form a clique. This not 
only yields a linear recognition algorithm for chordal graphs, but also a greedy coloring 
algorithm for finding a minimum coloring for this graph class~\cite{golumbic2004book}.
As most graph searching paradigms can be
implemented in linear time, these algorithms are typically as efficient as possible.

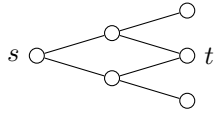
\begin{figure}
	\centering
	\begin{tikzpicture}[vertex/.style={inner sep=2pt,draw,circle},rotate=90, xscale=0.6]
	\node[vertex, label=right:$ t $] (1) at (0,0) {};
	\node[vertex] (2) at (-1,0) {};
	\node[vertex] (3) at (1,0) {};
	\node[vertex] (4) at (-0.5,1) {};
	\node[vertex] (5) at (0.5,1) {};
	\node[vertex, label=left:$ s $] (6) at (0,2) {};
	
	\draw[] (1) -- (4) -- (6) -- (5) -- (1);
  \draw[] (2) -- (4);
  \draw[] (3) -- (5);
	\end{tikzpicture}\caption{Vertex $t$ cannot be the end-vertex of any BFS although it is in the last layer if the BFS starts in $s$.}\label{fig:bfs}
\end{figure}

Interestingly, the end-vertices of graph searches, i.e. the last vertices visited in the search, are crucial for several algorithms. Their properties are the key for many multi-sweep algorithms on graphs. For instance, one can use six LBFS runs 
to construct the interval model of an interval graph. Here, the end-vertices correspond to the end vertices of the interval model 
and the next search starts in an end-vertex of the previous one. This also yields a linear time recognition algorithm for this 
graph class~\cite{corneil2009lbfs}.

Additionally, end-vertices may have strong structural properties. For example, as a direct consequence of the results mentioned above, the end-vertex of an LBFS on a chordal graph is always simplicial.
Moreover, if a cocomparability graph is hamiltonian, then the end-vertex of an LDFS is the start vertex of a hamiltonian path in this graph~\cite{corneil2013ldfs}. 

LBFS also provides a linear time algorithm for finding dominating pairs in connected asteroidal triple-free graphs~\cite{corneil99}. Here, a dominating pair is a pair of vertices such that every path connecting them is a dominating set in the graph. The first vertex $x$ is simply the end-vertex of an arbitrary LBFS and the second vertex $y$ is the end-vertex of an LBFS starting in $x$.
Furthermore, the end-vertex of an LBFS in a cocomparability graph is always a source/sink in some transitive
orientation of its complement~\cite{habib2000}.

The end-vertices of BFS are also helpful for fast diameter computation. Crescenzi et al.~\cite{crescenzi2013} have shown that the diameter of large real world graphs can usually be found with only a few applications of BFS. Intuitively, the  potential end-vertices can be seen as  peripheral or extremal vertices of the graph whereas all other vertices appear more centrally and are of greater importance to the connectivity. For example, no search can end on a cut vertex~\cite{charbit2014influence}. However, even for BFS finding such vertices is not as trivial as it might seem. Figure~\ref{fig:bfs} shows a graph in which the vertex $ t $ appears in the last layer of a BFS starting in $ s $. Nevertheless, it cannot be an end-vertex.

In this context, the decision problem arises, whether a vertex can be the \emph{end-vertex} of a graph search, i.e., the last vertex visited by this search.

\begin{problem}{End-vertex Problem}
	Instance: & A connected graph $ G=(V,E) $ and a vertex $ t \in V $. \\ 
	Task: & Decide whether there is a graph search such that $ t $ is the end-vertex of this search on $ G $.
\end{problem}

Obviously, this problem is in $ \NP $ for any of the searches considered here, since a full search order provides a certificate which can be checked in polynomial time. 

Although the problem of checking whether a given tree is a BFS or DFS tree can be decided efficiently (see~\cite{hagerup1985,korach1989,manber1990}), surprisingly this does not hold for the end-vertex problem. Corneil, K\"ohler, and Lanlignel~\cite{corneil2010end} have shown that it is \NP-hard to decide whether a vertex can be the end-vertex of an LBFS. Charbit, Habib, and Mamcarz~\cite{charbit2014influence} have shown that the end-vertex problems for BFS and DFS are also \NP-complete. Furthermore, they extended these results to several graph classes. In this paper, we address the end-vertex problem for MCS and MNS; for some other related results see Table \ref{tab:results}.

\paragraph{Our contribution}

In the following, we give an overview of graph searching algorithms, especially Maximum Cardinality Search and Maximal Neighborhood Search.
We present \NP-completeness results for the end-vertex problem of MNS on weakly chordal graphs and of MCS on general graphs. For some chordal graph classes we give linear time algorithms for several different graph searches. An overview of our results can be found in Table~\ref{tab:results}. We conclude the paper with some open problems related to our results.

\begin{table}[ht!]
	\centering\small
	\begin{tabular}{l c c c c c c}
		\toprule
                            &  BFS                              & LBFS                           & DFS                              & LDFS         & MCS       & MNS     \\ \midrule
		All Graphs              &  NPC                              & NPC                            & NPC                              & NPC          & \bfseries{NPC}   & \bfseries{NPC} \\ 
		Weakly Chordal          &  NPC \cite{charbit2014influence}  & NPC \cite{corneil2010end}      & NPC                              & NPC \cite{charbit2014influence}  & ?                & \bfseries{NPC}              \\
		Chordal                 &  ?                                & ?                              & NPC                              & ?            & ?                & P \cite{berry2010graph} $\to$ \textbf{L}              \\
		Interval                &  ?                                & L \cite{corneil2010end}        & \textbf{L}                                & ?            & ?                & P $\to$ \textbf{L}  \\
    Unit Interval           &  ?                       & L                              & \textbf{L}                       & \textbf{L}   & \bfseries{L}     & P $\to$ \textbf{L} \\
		Split                   &  L \cite{charbit2014influence}    & P \cite{charbit2014influence} $\to$ \textbf{L}  & NPC \cite{charbit2014influence}  & P \cite{charbit2014influence} $\to$ \textbf{L}  & \bfseries{L}   & P $\to$ \textbf{L}      \\ \bottomrule \addlinespace
	\end{tabular}
	\caption{Complexity of the end-vertex problem. Bolded results are from the presented paper. The big L stands for linear time algorithm. The term $P \to L$ describes the improvement from a polynomial algorithm to a linear time algorithm. Non-bolded results without references are direct consequences of other results, e.g., $\NPC$ of DFS on split graphs implies $\NPC$ on weakly chordal, chordal and general graphs. On the other hand, the linear time algorithm for interval graphs is also a linear time algorithm for the subclass of unit interval graphs.}\label{tab:results}
\end{table}

\section{Preliminaries}

All graphs considered in this article are finite, undirected, simple and connected. Given a graph $G=(V,E)$, we will denote by $n$ and $m$ the number of vertices and edges in $G$, respectively. Given a graph $G=(V,E)$ and a vertex $v\in V$, we will denote by $N(v)$ the \emph{neighborhood} of the vertex $v$, i.e., the set $N(v)=\{u\in V\mid uv\in E\}$, where an edge between $u$ and $v$ in $G$ is denoted by $uv\in E$. The \emph{closed neighborhood} of $v$ is the set $N[v]=N(v)\cup \{v\}$. A \emph{clique} in a graph $G$ is a set of pairwise adjacent vertices and an \emph{independent set} in $G$ is a set of pairwise nonadjacent vertices. If the neighborhood of a vertex $v$ in $G$ is a clique, then $v$ is said to be a \emph{simplicial vertex.}  The \emph{complement} of the graph $G$ is a graph $\overline{G}$ having the same set of vertices as $G$ where $xy$ for $x,y\in V$ is an edge of $\overline{G}$ if and only if it is not an edge in $G$.

Given a subset $S$ of vertices in $G$, we denote by $G[S]$ the \emph{subgraph of $G$ induced by $S$}, where $V(G[S])=S$ and $E(G[S])=\{xy\in E(G) \mid x\in S, y\in S\}$. By $G-S$ we denote the graph induced by $V(G)\setminus S$. By $P_n$ and $C_n$ we denote the path and the cycle on $n$ vertices, respectively. A \emph{claw} is a graph which consists of one vertex adjacent to exactly three vertices of degree one. A \emph{net} consists of a $C_3$, whose vertices are adjacent to three different vertices of degree one. For an arbitrary graph $H$, a graph $G$ is called $H$-free if it does not contain $H$ as an induced subgraph.

Given nonadjacent vertices $a$ and $b$ in graph $G$, a subset $C$ of vertices $V\setminus \{a,b\}$ is said to be an $(a,b)$-separator in $G$ if $a$ and $b$ are in distinct connected components of $G-C$. If no proper subset of $C$ is an $(a,b)$-separator in $G$, then $C$ is a minimal $(a,b)$-separator in $G$. Finally, $C$ is a minimal separator in $G$ if it is a minimal $(a,b)$-separator for some pair of nonadjacent vertices $a$ and $b$.

A graph $G$ that contains no induced cycle of length larger than $3$ is called \emph{chordal}. If neither $G$ nor its complement contain an induced cycle of length $5$ or more, then $G$ is \emph{weakly chordal}. A graph $G$ is said to be an \emph{interval graph} if its vertices can be represented as intervals on the real line, with two vertices being adjacent if and only if their corresponding intervals intersect. An interval graph $G$ is a \emph{unit interval graph} (or \emph{proper interval graph}) if all its intervals are of equal length. A split graph $G$ is a graph whose vertex set can be divided into sets $C$ and $I$ such that $C$ is a clique in $G$ and $I$ is an independent set in $G$. It is obvious that every unit interval graph is also an interval graph. Furthermore, it is easy to see that split graphs and interval graphs are chordal, whereas every chordal graph is also weakly chordal.

An ordering of vertices in $G$ is a bijection $\sigma: V(G) \rightarrow \{1,2,\dots,n\}$. For an arbitrary ordering $\sigma$ of vertices in $G$, we will denote by $\sigma(v)$ the position of vertex $v\in V(G)$. Given two vertices $u$ and $v$ in $G$ we say that $u$ is \emph{to the left} (resp. \emph{to the right}) of $v$ if $\sigma(u)<\sigma(v)$ (resp. $\sigma(u)>\sigma(v)$) and we will denote this by $u \prec_{\sigma}v$ (resp.  $u \succ_{\sigma}v$). A vertex $v$ is an \emph{end-vertex} of an ordering $\sigma$ of $G$ if $\sigma(v)=n$.

\section{Graph Searches}

In 1976 Rose, Tarjan and Lueker defined a linear time algorithm (Lex-P) which computes a perfect elimination ordering, if any exists, thus giving a recognition algorithm for chordal graphs~\cite{rose1976}. This algorithm, since named \emph{Lexicographic Breadth First Search}, or LBFS, exhibits many interesting structural properties and has been used as an ingredient in many other recognition and optimization algorithms. In the following we will present three graph searches which were derived from LBFS in different ways.

\emph{Maximum Cardinality Search} (MCS) was introduced in 1984 by Tarjan and Yannakakis~\cite{tarjan1984simple} as a simple alternative to LBFS for recognizing chordal graphs. They noticed that instead of remembering the \emph{order} in which previous neighbors of a vertex had appeared, it sufficed to just store the \emph{number} of previously visited neighbors for each vertex. This observation resulted in an algorithm which had a linear running time and which was very easy to implement.

\begin{algorithm}[ht!]
	\KwIn{Connected graph $G=(V,E)$ and a distinguished vertex $ s \in V $}
	\KwOut{A vertex ordering $ \sigma $}
	\Begin{	
		$\sigma(1)  \leftarrow  s $\;	
		\For{$ i \leftarrow 2 $ to $ n $}{pick an unnumbered vertex $ v $ with largest amount of numbered neighbors\;
			$ \sigma(i) \leftarrow v$\;}		
	}\caption{Maximum Cardinality Search}
	\label{mcs}
\end{algorithm}

In~\cite{corneil2008unified}, Corneil and Krueger defined \emph{Lexicographic Depth First Search} as a lexicographic analogue to DFS. Since then, it has been used for many applications, most notably to solve the minimum path cover problem on cocomparability graphs~\cite{corneil2013ldfs}.

\begin{algorithm}[ht!]
	\KwIn{Connected graph $G=(V,E)$ and a distinguished vertex $ s \in V $}
	\KwOut{A vertex ordering $ \sigma $}
	\Begin{
		$ label(s) \leftarrow \{0\} $\;
		
		\lFor{each vertex $v \in V-{s}$}{assign to $ v $ the empty label}
		
		\For{$ i \leftarrow 1 $ to $ n $}{pick an unnumbered vertex $ v $ with lexicographically largest label\;
			$ \sigma(i) \leftarrow v$\;
			\lForEach{unnumbered vertex $ w \in N(v) $}{prepend $ i $ to $ label(w) $}}		
	}\caption{LDFS}\index{LDFS}
	\label{ldfs}
\end{algorithm}

\emph{Maximal Neighborhood Search (MNS)} was introduced by Corneil and Krueger~\cite{corneil2008unified} in 2008 as a generalization of LBFS, LDFS and MCS. Instead of using strings (like LBFS) or integers (like MCS) the algorithm uses sets of integers as labels and the maximal labels are those sets which are inclusion maximal. Unlike the labels of LBFS and MCS the labels of MNS are not totally ordered and there can be many different maximal labels. Corneil and Krueger showed that every search ordering of LBFS and MCS is also an MNS ordering. This result was generalized in 2009 by Berry et al.~\cite{Berry2009mls} who showed that MNS orderings even contain all orderings of a generalization of the known graph searches called \emph{Maximal Label Search}.

\begin{algorithm}[ht!]
	\KwIn{Connected graph $G=(V,E)$ and a distinguished vertex $ s \in V $}
	\KwOut{A vertex ordering $ \sigma $}
	\Begin{
		$ label(s) \leftarrow  \{n+1\} $\;
		\lFor{each vertex $v \in V-{s}$}{assign to $ v $ the empty label}
		\For{$ i \leftarrow 1 $ to $ n $}{pick an unnumbered vertex $ v $ with maximal label under set inclusion\;
			$ \sigma(i) \leftarrow v$\;
			\lForEach{unnumbered vertex $ w $ adjacent to $ v $}{add $ i $ to $ label(w) $}	
		}	
	}\caption{Maximal Neighborhood Search}
	\label{algo:mns}
\end{algorithm}

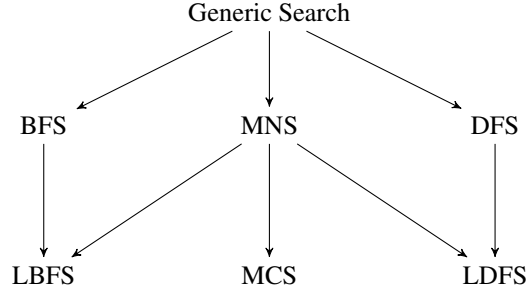
\begin{figure}[ht!]
	\centering
	\begin{tikzpicture}
	\node (gs) at (3,5.5) {Generic Search};%
	\node (bfs) at (0,4) {BFS};%
	\node (dfs) at (6,4) {DFS};%
	\node (mns) at (3,4) {MNS};%
	\node (mcs) at (3,2) {MCS};%
	\node (lbfs) at (0,2) {LBFS};%
	\node (ldfs) at (6,2) {LDFS};%
	\draw[-stealth'] (gs)--(bfs);%
	\draw[-stealth'] (gs)--(mns);%
	\draw[-stealth'] (gs)--(dfs);%
	\draw[-stealth'] (bfs)--(lbfs);%
	\draw[-stealth'] (mns)--(mcs);%
	\draw[-stealth'] (dfs)--(ldfs);%
	\draw[-stealth'] (mns)--(lbfs);%
	\draw[-stealth'] (mns)--(ldfs);%
	\end{tikzpicture}\caption{This figure represents a summary of the relationships between graph searches. The arrows represent proper inclusions. Thus, for example the arrow between BFS and LBFS implies that every LBFS is also a BFS. Searches on the same level are not comparable.}
\end{figure}

\section{\NP-Completeness for Maximal Neighborhood Search}

The complexity of the end-vertex problem of MNS was studied by Berry et al.~\cite{berry2010graph} in 2010 resulting in the following characterization.

\begin{lemma}\cite{berry2010graph}\label{mns:lemma1}
	Let $ G=(V,E) $ be a chordal graph and let $ t \in V $. Then $ t $ can be an end-vertex of MNS if and only if $t$ is simplicial and the minimal separators included in $N(t)$ are totally ordered by inclusion.
\end{lemma}

Since this property can be checked efficiently, they conclude that the end-vertex problem of MNS on chordal graphs is solvable in polynomial time. In Section~\ref{sec:mcs_split} we provide an approach which solves this problem in linear time (see Corollary~\ref{polynomial:lemma1}).

Charbit et al.~\cite{charbit2014influence} conjectured that the problem can be solved efficiently on general graphs. However, in this section we will present an \NP-completeness proof for the end-vertex-problem of MNS on weakly chordal graphs.

\begin{theorem}\label{theo:mns_npc}
The end-vertex-problem of MNS is \NP-complete for weakly chordal graphs.
\end{theorem}

To prove this we use a reduction from 3-SAT. Let $\cal{I}$ be an instance of 3-SAT. We construct the corresponding graph $G(\mathcal{I})$ as follows (see Figure~\ref{fig:mns_end_vertex} for an example). Let $ X=\{x_1, \dots, x_k,\overline{x}_1,\ldots,\overline{x}_k\} $ be the set of vertices representing the literals of $ \mathcal{I} $. The edge-set $ E(X) $ forms the complement of the matching in which $x_i$ is matched to $ \overline{x}_i $ for every $ i \in \{1,\ldots, k\} $. Let $ C = \{c_1, \ldots ,c_l\} $ be the set of clause vertices representing the clauses of $ \mathcal{I} $. The set $ C $ is independent in $G(\mathcal{I})$ and every $ c_i $ is adjacent to every vertex of $ X $, apart from those representing the literals of the clause associated with $ c_i $ for every $ i \in \{1, \ldots ,l \} $. Additionally, we add the vertices $s$, $b$ and $t$. The vertex $b$ is adjacent to all literal vertices. The vertices $s$ and $t$ are adjacent to all literal and all clause vertices. Finally, we add the edge $bt$.

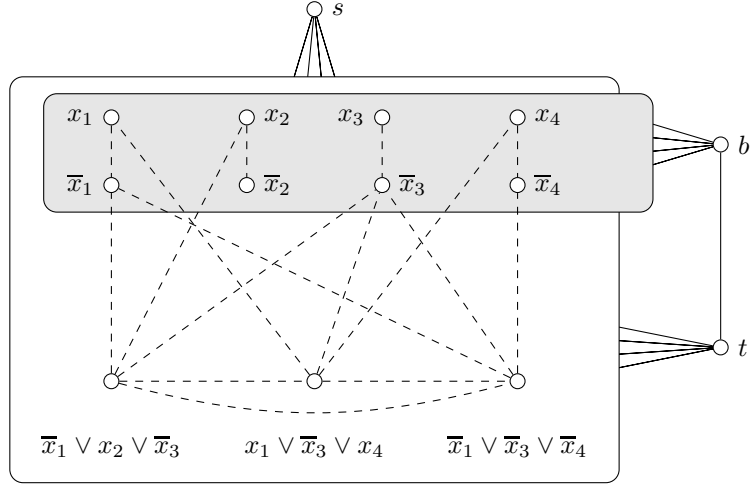
\begin{figure}
	\centering
 	\begin{tikzpicture}[scale=0.90, vertex/.style={inner sep=2pt,draw,circle, fill=white},
 	noedge/.style={dashed}
 	]
 	
 	\draw[rounded corners=5pt]
 	(-2.5,3) rectangle (6.5,9);
 	\draw[rounded corners=5pt, fill=black!10!white]
 	(-2,7) rectangle (7,8.75);

 	\node[vertex,label={0:$s$}] (s) at (2,10) {};
 	\node[vertex,label={180:$x_1$}] (x1) at (-1, 8.4) {};
 	\node[vertex,label={180:$\overline{x}_1$}] (nx1) at (-1, 7.4) {};
 	\node[vertex,label={0:$x_2$}] (x2) at (1, 8.4) {};
 	\node[vertex,label={0:$\overline{x}_2$}] (nx2) at (1, 7.4) {};
 	\node[vertex,label={180:$x_3$}] (x3) at (3, 8.4) {};
 	\node[vertex,label={0:$\overline{x}_3$}] (nx3) at (3, 7.4) {};
 	\node[vertex,label={0:$x_4$}] (x4) at (5, 8.4) {};
 	\node[vertex,label={0:$\overline{x}_4$}] (nx4) at (5, 7.4) {};
 	
 	\node[vertex,label={[label distance=0.5cm]270:$\overline{x}_1 \lor x_2 \lor \overline{x}_3$}] (c1) at (-1, 4.5) {};
 	\node[vertex,label={[label distance=0.5cm]270:$x_1 \lor \overline{x}_3 \lor x_4$}] (c2) at (2,4.5) {};
 	\node[vertex,label={[label distance=0.5cm]270:$\overline{x}_1 \lor \overline{x}_3 \lor \overline{x}_4$}] (c3) at (5, 4.5) {};
 	
 	\node[vertex,label={0:$t$}] (t) at (8,5) {};
 	\node[vertex,label={0:$b$}] (b) at (8,8) {};
 	
 	\draw[noedge] (x1) -- (nx1);
 	\draw[noedge] (x2) -- (nx2);
 	\draw[noedge] (x3) -- (nx3);
 	\draw[noedge] (x4) -- (nx4);

 	\draw[noedge] (c1) -- (nx1);
 	\draw[noedge] (c1) -- (x2);
 	\draw[noedge] (c1) -- (nx3);
 	\draw[noedge] (c2) -- (x1);
 	\draw[noedge] (c2) -- (nx3);
 	\draw[noedge] (c2) -- (x4);
 	\draw[noedge] (c3) -- (nx1);
 	\draw[noedge] (c3) -- (nx3);
 	\draw[noedge] (c3) -- (nx4);
 	
 	\draw[noedge] (c1) -- (c2);
 	\draw[noedge] (c2) -- (c3);
 	\draw[noedge, bend angle=15, bend right] (c1) to (c3);

 	\draw[] (s) -- (1.7,9) -- (s)--(2.1,9)--(s)--(2.3,9)--(s)--(1.9,9);
 	
 	\draw[] (t) -- (6.5,4.9)--(t)--(6.5,4.7)--(t)--(6.5,5.1)--(t)--(6.5,5.3);
 	
 	\draw[] (b) -- (7,8.1) --(b)--(7,7.9)--(b)--(7,7.7)--(b)--(7,8.3);
 	
 	\draw[] (b) -- (t);
 	
 	\end{tikzpicture}
 	\caption{The \NP-completeness reduction for the end-vertex problem of MNS on weakly chordal graphs. The depicted graph is $G(\mathcal{I})$ for $\mathcal{I} = (\overline{x}_1 \lor x_2 \lor \overline{x}_3) \land (x_1 \lor \overline{x}_3 \lor x_4) \land (\overline{x}_1 \lor \overline{x}_3 \lor \overline{x}_4)$. In both boxes only non-edges are displayed by dashed lines. The connection of a vertex with a box means that the vertex is adjacent to all vertices in the box.}\label{fig:mns_end_vertex}
\end{figure}

The following two lemmas provide some properties that an MNS ordering must fulfill if $t$ is its end-vertex. Let $\mathcal{I}$ be an arbitrary instance of $3$-SAT and $G=G(\mathcal{I})$ with $n=|V(G)|$. 

\begin{lemma}\label{lemma:mns_ordering_bsclause} 
Let $\sigma$ be an MNS ordering of vertices in $G$, ending with $t$. Then:
\begin{enumerate}
\item The vertex $b$ is on the left of any clause vertex in ordering $\sigma$;
\item The vertex $s$ is on the left of $b$ in ordering $\sigma$.
\end{enumerate}
\end{lemma}

\begin{proof}
We prove both claims separately. Assume that there is a clause vertex $c_i$ that is visited before the vertex $b$, i.e., $c_i \prec_{\sigma} b$. By the construction we know that $N(b)\subset N(t)$ and $C\subseteq N(t)$ while $N(b)\cap C=\emptyset$. If vertex $c_i$ is visited in $\sigma$ before the vertex $b$, then the label of vertex $t$ will contain the label $c_i$, while for $b$ this is not the case. This implies that at any step of the search process in $G$, the label set of $b$ will be a proper subset of the label set of $t$, and the algorithm will take vertex $t$ before $b$.

Assume now that the vertex $b$ is on the left of $s$ in the ordering $\sigma$. Since $N(s) \subset N(t)$ and $b$ is in the neighborhood of $t$ and $s$ is not, the label set of $s$ will be a proper subset of the label set of $t$. Thus, the search algorithm will visit $t$ before $s$.
\end{proof}

\begin{lemma}\label{lemma:mns_assignment}
Let $\sigma$ be an MNS ordering of vertices in $G$, ending in $t$. Then the first $k+1$ vertices in $\sigma$ are $s$, as well as an arbitrary assignment of $\I$ (not necessarily satisfying).
\end{lemma}

\begin{proof}
It follows from Lemma~\ref{lemma:mns_ordering_bsclause} that the first vertex in $\sigma$ is $s$ or one of the literal vertices. Assume that we are at some step of an MNS search and until now we have only chosen at most one literal vertex per variable, possibly including the vertex $s$. Further assume that there is at least one variable whose two literal vertices have not been chosen so far. These vertices are adjacent to all the vertices which have been chosen before. Lemma~\ref{lemma:mns_ordering_bsclause} implies that a clause vertex cannot be next in the MNS ordering. For the same reason, if $s$ has not been visited so far, we cannot take $b$ next. If $s$ has already been chosen, then the labels of literal vertices of the non-chosen variables are proper supersets of the label of $b$; again we cannot choose the vertex $b$. The same holds for literal vertices of variables, for which the other literal vertex has already been visited. Thus, we know that we have to take $s$ or a literal vertex of an unvisited variable. Notice that each of these vertices is a possible choice in an MNS, since they all have maximum label. This proves the statement. 
\end{proof}

So far we have proven some necessary conditions that have to be satisfied in order for $t$ to be an MNS end-vertex in $G$. In what follows, we will prove that $t$ can be the end-vertex of MNS in $G$ whenever the corresponding $3$-SAT instance $ \I $ has a satisfying assignment.

\begin{lemma}\label{lemma:mns_satisfying}
If the 3-SAT instance $\I$ has a satisfying assignment $\A$, then $ t $ is an end-vertex of MNS on $ G $.
\end{lemma}

\begin{proof}
Let $\I$ be an instance of $3$-SAT and $\A$ a satisfying assignment of $\I$. We now construct the MNS ordering which ends in $t$.

From Lemma~\ref{lemma:mns_assignment} it follows that we can start an MNS search on $G$ in the vertex $s$ and then take all the literal vertices that belong to $\A$. Since each clause vertex is not adjacent to its corresponding literals, it follows that for each clause vertex there is at least one label missing among the labels produced by the assignment vertices. The same holds for the unvisited literal vertices, since their negated literal label is missing. Furthermore, since $b$ is adjacent to all literal vertices, labels of unvisited literal vertices as well as of clause vertices do not contain the label set of vertex $b$. Therefore, we can take $b$ as the next vertex. Since all remaining literal vertices and clause vertices are adjacent to $s$, while $t$ is not, these remaining vertices can be visited before $t$.
\end{proof}

We now show that $t$ cannot be the end-vertex of an MNS if there is no satisfying assignment of \I.

\begin{lemma}\label{lemma:mns_non_satisfying}
If the 3-SAT instance $\I$ has no satisfying assignment, then $t$ cannot be the end-vertex of MNS on $G$.
\end{lemma}

\begin{proof}
Let $\I$ be an instance of $3$-SAT which does not have a satisfying assignment. Suppose that there exists an MNS ordering $\sigma$ of $G$ which ends in $t$. It follows from Lemma~\ref{lemma:mns_assignment} that MNS has to take the vertex $s$ and an arbitrary assignment of $\I$ at the beginning. Since the assignment that was chosen is not satisfying, there is at least one clause vertex that has been labeled by all vertices chosen so far. Hence, the labels of these clause vertices properly contain the labels of all remaining literal vertices, as well as the labels of $b$ and $t$. As a result, MNS has to take one of these clause vertices next. Lemma~\ref{lemma:mns_ordering_bsclause} implies that then $t$ cannot be an end-vertex, since one clause vertex was chosen before $b$. 
\end{proof}

To complete the proof of Theorem~\ref{theo:mns_npc} we have to show, that $G(\I)$ is weakly chordal for any choice of $\I$.

\begin{lemma}\label{lemma:mns_weakly_chordal}
The graph $G(\mathcal{I})$ is weakly chordal for any instance $ \I $ of 3-SAT.
\end{lemma}

\begin{proof}
By contradiction assume that there exists an induced cycle of length $\geq 5$ in $G(\mathcal{I})$ or its complement. We start with $G(\mathcal{I})$. Vertex $t$ cannot be part of such a cycle, since there is only one vertex which is not adjacent to $t$. Thus, $s$ also cannot be part of such a cycle, since it is adjacent to all vertices but $t$ and $b$. 
Note, if the cycle contains four or more vertices that are not clause vertices there must be a chord in the cycle. Therefore, the cycle contains at least two clause vertices $c_i$ and $c_j$. The neighbors of $c_i$ and $c_j$ in the cycle are literal vertices. If the neighbors are four different vertices, then there exists at least four edges between them. Therefore, there must be a chord. If both share one neighbor then there is also a chord between these shared neighbor and one of the other two. The clause vertices cannot share both neighbors, since the cycle has more than four vertices. Thus, there is no induced cycle with more than four vertices in $G(\mathcal{I})$ .

Consider now $\overline{G(\mathcal{I})}$. Since $t$ is only adjacent to $s$, it is not part of a cycle. The same holds for $s$, since it is only adjacent to $t$ and $b$. As the clause vertices build a clique, there are at most two of them in a cycle which must be consecutive. Between literal vertices of different variables and between $b$ and a literal vertex must lie at least one clause vertex in the cycle. This leads to a contradiction, since we need at least two non-consecutive clause vertices.
\end{proof}

Theorem~\ref{theo:mns_npc} follows from the Lemmas~\ref{lemma:mns_satisfying},~\ref{lemma:mns_non_satisfying} and~\ref{lemma:mns_weakly_chordal}.


\section{\NP-Completeness for Maximum Cardinality Search}

To the best of our knowledge, the end-vertex problem for Maximum Cardinality Search has not been studied in the literature. It is easy to see that for trees the end-vertices correspond exactly to the leaves. 

However, in this section, we will prove that the end-vertex problem is $\NP$-complete on general graphs, by giving a reduction from 3-SAT.

\begin{theorem}\label{mcs-end-vertex}
	The MCS end-vertex problem is $\NP$-complete.
\end{theorem}

For each instance $ \I $ of 3-SAT we construct a corresponding graph $G(\I)$ as follows (see Figures~\ref{mcs:figure2}, \ref{mcs:figure1} and \ref{mcs:figure5} for an example): Each literal is represented by an edge, and each clause by a triangle. The triangles representing the clauses together form a clique $ C $ of size $ 3l $. We also define start vertices $ s $ and $ s' $, where $ s $ is adjacent to all vertices of $ x_1  $ and $ \overline{x}_1 $ and $ s' $ is just adjacent to $ s $. Two consecutive literals $ x_{j-1} $ and $ x_{j} $ are connected using two auxiliary vertices in the way described in Figure~\ref{mcs:figure1}. For each literal one of its two vertices is adjacent to all three vertices of each clause which contains the negation of that literal, as depicted in Figure~\ref{mcs:figure5}.

Additionally, the graph contains a clique $ K $ with $3(4k+8(k-1)) + 4$ vertices, i.e., three vertices for each literal and auxiliary vertex, as well as 4 connector vertices. Every vertex among the literal vertices and the auxiliary vertices is adjacent to exactly 3 vertices in $ K $ such that every vertex in $ K $ apart from the 4 connector vertices are adjacent to exactly one vertex outside of $ K $. Two of the connector vertices are then completely connected to all vertices of $ x_k $ and the other two to all vertices of $ \overline{x}_k $.

Finally, we add a vertex $ t $, for which we wish to decide whether it is an end-vertex, and this is adjacent to all clause vertices.

\begin{figure}
	\centering
	\begin{tikzpicture}[
	edge/.style={-},
	noedge/.style={dashed},
	tree/.style={snake}, auxiliary/.style={inner sep =1pt,draw,circle},
	vertex/.style={inner sep=2pt,draw,circle},
	dvertex/.style={draw,ellipse,minimum height=30,minimum width=10,fill=white},
	auxiliary/.style={inner sep =1pt,draw,circle,fill=white},
	cvertex/.style={draw,rectangle,rounded corners=4pt,minimum height=10,minimum width=30,fill=gray,fill opacity=0.2}
	]
	
	
	\draw (5.75,6.5) ellipse (5cm and 1.15cm);

	\draw[rounded corners=5pt,fill=gray,fill opacity=0.2]
	(1.25,1) rectangle (10.25,5);	
	
	\coordinate (f0) at (10.25,3);
	\coordinate (f1) at (11.25,3.5);
	\coordinate (f2) at (11.25,3);
	\coordinate (f3) at (11.25,2.5);
	\draw (f0) -- (f1);
	\draw (f0) -- (f2);
	\draw (f0) -- (f3);	
	
	\draw[draw,fill=white] (12,3) ellipse (1cm and 1.5cm);
	\node[] (C) at (12,3) {$K$};
	
	\node[vertex,label=above:$s'$] (s) at (0,3) {};
	\node[vertex, label=above:$s$] (0) at (0.7,3) {};
	
	
	\coordinate (c1) at (2,4) {};
	\coordinate (c2) at (2,2) {};
	\coordinate (c3) at (3.5,4) {};
	\coordinate (c4) at (3.5,2) {};
	\coordinate (c5) at (5,4) {};
	\coordinate (c6) at (5,2) {};
	\coordinate (c7) at (8,4) {};
	\coordinate (c8) at (8,2) {};
	\coordinate (c9) at (9.5,4) {};
	\coordinate (c10) at (9.5,2) {};

	\draw[double,thick] (c3)--(c1)--(c4)--(c6)--(c3);
	\draw[double,thick] (c4)--(c2)--(c3)--(c5)--(c4);
	\draw[double, thick] (c10)--(c7);
	\draw[double, thick] (c7)--(c9)--(c8)--(c10);

	\node[dvertex,label=above:$x_1$] (1) at (2,4) {};
	\node[dvertex,label=below:$\overline{x}_1$] (2) at (2,2) {};

	\node[dvertex,label=above:$x_2$] (3) at (3.5,4) {};
	\node[dvertex,label=below:$\overline{x}_2$] (4) at (3.5,2) {};
	
	\node[dvertex,label=above:$x_3$] (5) at (5,4) {};
	\node[dvertex,label=below:$\overline{x}_3$] (6) at (5,2) {};
	
	\node[] (58) at (6.5,4) {$ \ldots $};
	\node[] (59) at (6.5,2) {$ \ldots $};
	
	\node[dvertex,label=above:$x_{k-1}$] (7) at (8,4) {};
	\node[dvertex,label=below:$\overline{x}_{k-1}$] (8) at (8,2) {};
	
	\node[dvertex,label=above:$x_{k}$] (9) at (9.5,4) {};
	\node[dvertex,label=below:$\overline{x}_{k}$] (10) at (9.5,2) {};
	
	\node[auxiliary] (11) at (2,4.2) {};
	\node[auxiliary] (12) at (2,3.8) {};
	
	\node[auxiliary] (13) at (2,2.2) {};
	\node[auxiliary] (14) at (2,1.8) {};
	
	\node[auxiliary] (15) at (9.5,4.2) {};
	\node[auxiliary] (16) at (9.5,3.8) {};
	
	\node[auxiliary] (17) at (9.5,2.2) {};
	\node[auxiliary] (18) at (9.5,1.8) {};
	
	\node[auxiliary] (19) at (11.5,4) {};
	\node[auxiliary] (20) at (11.4,3.6) {};
	
	\node[auxiliary] (21) at (11.4,2.4) {};
	\node[auxiliary] (22) at (11.5,2) {};
	
	
	\node[cvertex,label=below left:$c_{1}$] (23) at (2,6.5) {};
	\node[cvertex,label=below left:$c_{2}$] (24) at (3.5,6.5) {};
	\node[cvertex,label=below left:$c_{3}$] (25) at (5,6.5) {};
	\node[cvertex,label=below left:$c_{l-1}$] (26) at (8,6.5) {};
	\node[cvertex,label=below left:$c_{l}$] (27) at (9.5,6.5) {};

	\node[] (60) at (6.5,6.5) {$ \ldots $};
	
	
	\node[auxiliary] (28) at (1.8,6.5) {};
	\node[auxiliary] (29) at (2,6.5) {};
	\node[auxiliary] (30) at (2.2,6.5) {};
	
	\node[auxiliary] (31) at (3.3,6.5) {};
	\node[auxiliary] (32) at (3.5,6.5) {};
	\node[auxiliary] (33) at (3.7,6.5) {};
	
	\node[auxiliary] (34) at (4.8,6.5) {};
	\node[auxiliary] (35) at (5,6.5) {};
	\node[auxiliary] (36) at (5.2,6.5) {};

	\node[auxiliary] (37) at (7.8,6.5) {};
	\node[auxiliary] (38) at (8,6.5) {};
	\node[auxiliary] (39) at (8.2,6.5) {};

	\node[auxiliary] (40) at (9.3,6.5) {};
	\node[auxiliary] (41) at (9.5,6.5) {};
	\node[auxiliary] (42) at (9.7,6.5) {};
	
	
	\coordinate (43) at (1.8,5.0);
	\coordinate (44) at (2,5.0);
	\coordinate (45) at (2.2,5.0);
	
	\coordinate (46) at (3.3,5.0);
	\coordinate (47) at (3.5,5.0);
	\coordinate (48) at (3.7,5.0);
	
	\coordinate (49) at (4.8,5.0);
	\coordinate (50) at (5,5.0);
	\coordinate (51) at (5.2,5.0);
	
	\coordinate (52) at (7.8,5.0);
	\coordinate (53) at (8,5.0);
	\coordinate (54) at (8.2,5.0);
	
	\coordinate (55) at (9.3,5.0);
	\coordinate (56) at (9.5,5.0);
	\coordinate (57) at (9.7,5.0);
	
	\node[vertex,label=above:$t$] (t) at (6,7) {};
	
	\draw[] (s)--(0)--(11)--(12)--(0)--(13)--(14)--(0);

	\draw[] (15)--(16)--(19)--(15)--(20)--(16)--(20)--(19);
	\draw[] (17)--(18)--(21)--(17)--(22)--(18)--(22)--(21);
	
	
	\draw[] (43)--(23)--(44)--(23)--(45);
	\draw[] (46)--(24)--(47)--(24)--(48);
	\draw[] (49)--(25)--(50)--(25)--(51);
	\draw[] (52)--(26)--(53)--(26)--(54);
	\draw[] (55)--(27)--(56)--(27)--(57);
	
	\end{tikzpicture}
	\caption{This represents a general construction of $ G(\I) $ for an arbitrary instance $ \I $. The double edges denote a construction linking the various literals and are explained in Figure~\ref{mcs:figure1}}\label{mcs:figure2}
\end{figure}
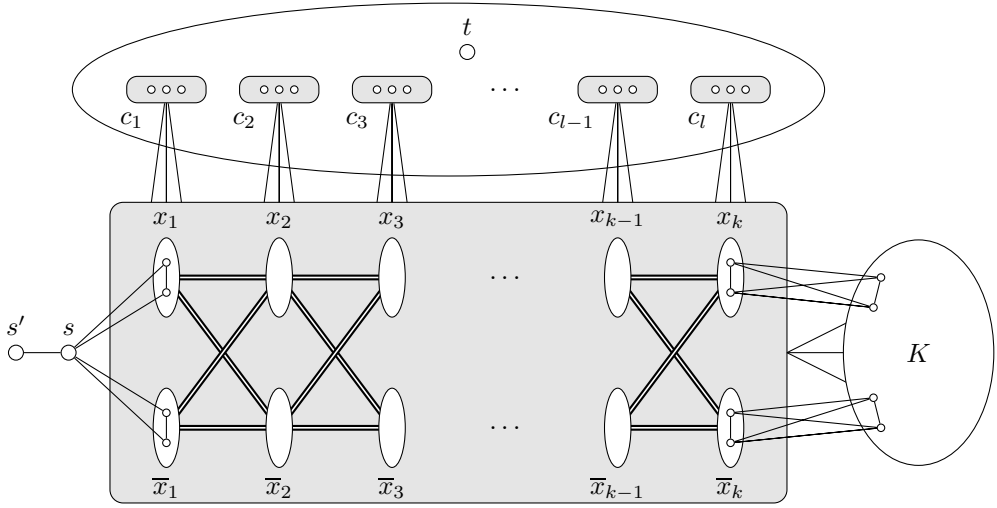

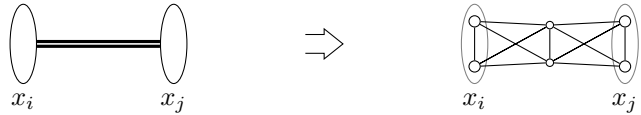
\begin{figure}
	\centering
	\begin{tikzpicture}[
		vertex/.style={inner sep=1.5pt,draw,circle},
		dvertex/.style={draw,ellipse,minimum height=30,minimum width=10},
		auxiliary/.style={inner sep =1pt,draw,circle}
		]
		
		\node[dvertex,label=below:$x_i$] (1) at (0,0) {};
		\node[dvertex,label=below:$x_j$] (2) at (2,0) {};
		
		\draw[double, very thick] (1)--(2);

	\begin{scope}[xshift=4cm]	
  \tikzarrow{black}{0}{0}{1}
	\end{scope}

	\begin{scope}[xshift=6cm]
	
	\node[dvertex,label=below:$x_i$,opacity=0.5] (1) at (0,0) {};
	\node[dvertex,label=below:$x_j$,opacity=0.5] (2) at (2,0) {};
	
	\node[vertex] (1) at (0,0.3) {};
	\node[vertex] (2) at (0,-0.3) {};
	
	\node[vertex] (3) at (2,0.3) {};
	\node[vertex] (4) at (2,-0.3) {};
		
	\node[auxiliary] (5) at (1,0.25) {};
	\node[auxiliary] (6) at (1,-0.25) {};
		
	\draw[] (1)--(2)--(5)--(1)--(6)--(2)--(5)--(6)--(3)--(5)--(4)--(6)--(3)--(4);
	
	\end{scope}
	
	\end{tikzpicture}
	\caption{Each double edge in the construction shown in Figure~\ref{mcs:figure2} is to be replaced by the construction on the right-hand side. Furthermore, each node, in particular the auxiliary vertices in the middle, is assigned to three exclusive neighbors in the clique $K$.}\label{mcs:figure1}
\end{figure}

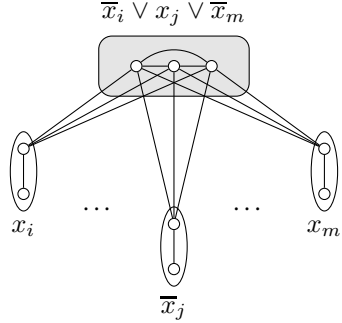
\begin{figure}
	\centering
	\begin{tikzpicture}[
		vertex/.style={inner sep=1.5pt,draw,circle,fill=white},
		dvertex/.style={draw,ellipse,minimum height=30,minimum width=10},
		auxiliary/.style={inner sep =1pt,draw,circle}]
	
	\begin{scope}[yshift=1cm]
	 \node at (0,1.1) {$\overline{x}_i \vee x_j \vee \overline{x}_m$};
	\draw[rounded corners=5pt,fill=gray,fill opacity=0.2]
	(-1,0) rectangle (1,0.8);	
	\node[vertex] (a) at (-0.5,0.4) {};
	\node[vertex] (b) at (0,0.4) {};
	\node[vertex] (c) at (0.5,0.4) {};
	\end{scope}

	\begin{scope}[xshift=-2cm,yshift=0cm]
	\node[dvertex,label=below:$x_i$] (1) at (0,0) {};
	\node[vertex] (11) at (0,-0.3) {};
	\node[vertex] (12) at (0, 0.3) {};
	\end{scope}
	
	\begin{scope}[xshift=-0cm,yshift=-1cm]
	\node[dvertex,label=below:$\overline{x}_j$] (2) at (0,0) {};
	\node[vertex] (21) at (0,-0.3) {};
	\node[vertex] (22) at (0, 0.3) {};
	\end{scope}
		
	\begin{scope}[xshift=2cm,yshift=0cm]
	\node[dvertex,label=below:$x_m$] (3) at (0,0) {};
	\node[vertex] (31) at (0,-0.3) {};
	\node[vertex] (32) at (0, 0.3) {};
	\end{scope}

	\node at (-1,-0.5) {\dots};
	\node at (1,-0.5) {\dots};

	\draw (11) -- (12);
	\draw (21) -- (22);
	\draw (31) -- (32);
	
	\draw (12) -- (a);
	\draw (12) -- (b);
	\draw (12) -- (c);
	\draw (22) -- (a);
	\draw (22) -- (b);
	\draw (22) -- (c);
	\draw (32) -- (a);
	\draw (32) -- (b);
	\draw (32) -- (c);
	
	\draw (a) -- (b) -- (c) [bend right=40] to (a);

	\end{tikzpicture}
	\caption{Each clause block consists of three nodes and each of these nodes is connected to one specific node of the corresponding literals.}\label{mcs:figure5}
\end{figure}

\begin{lemma}
	If $ \I $ has a satisfying assignment $ \A $, then there is a maximum cardinality search on $ G(\I) $ that ends in $ t $.
\end{lemma}
\begin{proof}
	Let $ \A = (b_1,b_2, \ldots, b_l)$ be a satisfying assignment of $ \I $. We can construct a corresponding MCS order $ \sigma $ ending in $ t $ as follows: As $ \sigma(1) $ we use the root $ s' $ and $ \sigma(2) = s $. If $ b_1=1 $ we choose the vertices of $ x_1 $ next; if $ b_1 = 0 $ we choose the vertices of $ \overline{x}_1 $. In the next step we use the construction described in Figure~\ref{mcs:figure1} to choose either $ x_2 $ or $ \overline{x}_2 $, depending on the value of $ b_2 $. We proceed in this way, until we have visited every literal given by $ \A $, that is until we reach $ x_k $ or $ \overline{x}_k $. This is possible, because we have chosen a satisfying assignment and the label of every clause vertex is at most 2, while there is always an auxiliary vertex or a literal vertex with label that is at least 2.
	
	At the point where we have visited both vertices in the $ k $-th literal, one pair of connector vertices of the clique $ K $ will have label 2. As we have chosen a fulfilling assignment in the manner described above, those connector vertices have maximal label at this point.
	
	Hence, it is possible to visit the whole clique $ K $ next. Note that after every vertex in $ K $ has been chosen, each vertex that has not been visited and that is neither a clause vertex nor $ t $ has a label larger or equal to 3, as each has three neighbors in the clique. On the other hand, every clause vertex can only have label at most 3, unless $ t $ is visited, as it has only three neighbors in $ G - (C\cup\{t\}) $. Thus, we can visit every vertex apart from the clause vertices and $ t $ first, then choose all the clause vertices in any order possible and finally visit $ t $ last.
\end{proof}

\begin{lemma}
	If $ \I $ does not admit a satisfying assignment, then $ t $ cannot be an MCS end-vertex of $ G(\I) $.
\end{lemma}
\begin{proof}
	Suppose that $ \I $ does not admit a satisfying assignment and there is an MCS ending in $ t $. First observe that any maximum cardinality search that ends in $ t $ must begin in $ s $ or $ s' $. If we start in an arbitrary vertex that is neither $ s $ nor $ s' $, then at any point of the search there will always be a vertex with label larger or equal to 2 until $ s' $ is the only vertex left. Without loss of generality, it begins in $ s' $, as otherwise we can choose $ s' $ as the second vertex. It is enough to show that it is not possible to reach $ K $ before visiting vertex $ t $. First, note that $ K $ can only be entered through the connector vertices, which in turn are only adjacent to $ x_k $ or $ \overline{x}_k $, as at any point of the MCS there is a vertex with label $ \geq 2 $, while a vertex in $ K $ has label $ \leq 1 $ until the connector vertices are chosen.
	
	Suppose that at any point in the search a clause vertex is chosen for the first time before we have reached $ K $. Then either all of the vertices of that clause are visited consecutively, or the remaining vertices of that clause will be labeled with 3 before we reach $ K $ (they must have had label 2 before the clause was entered and the choice of one clause vertex increased the label to 3). In this case, therefore, $ t $ cannot be chosen last in the search, as at least three clause vertices, and as a consequence also $ t $, must be chosen before we can enter $ K $.
	
	It remains to be shown, that a clause vertex has to be visited before entering $ K $. We have already argued that one of the literals of $ x_k $ must be visited before we can enter $ K $, as we can only enter through the connector vertices. If we disregard the clause vertices and $ K $, then the vertex set corresponding to a given variable $ x_i $ with $ i<k $ forms a separator between $ s $ and the vertices corresponding to the variable $ x_k $. Therefore, any search must traverse at least one vertex of each variable. Furthermore, as soon as the first vertex belonging to a literal is chosen, the other vertex in that literal has the largest label and must be chosen next. Hence, at least one literal for every variable must be visited, before the search reaches $ x_k $ (note that it is possible to visit both literals of a variable). As a result, we need to traverse \emph{at least} a whole assignment $ \A $ of the literals before entering $ K $. However, as $ \I $ does not admit a satisfying assignment, this implies that by the time we have visited one literal of $ x_k $, at least one of the clause vertices must have label 3, as all its literals have been assigned in the negative, and we must visit this vertex before entering the clique $ K $. 
	
	Therefore, it is impossible to reach $ K $ before we have chosen a clause vertex and $ t $ cannot be an end-vertex of an MCS.
\end{proof}

The following corollary concludes the proof of Theorem \ref{mcs-end-vertex}.

\begin{corollary}
	Let $\I$ be an instance of 3-SAT. Then $ \I $ has a satisfying assignment if and only if $ t $ is a possible MCS end-vertex of $ G(\I) $.
\end{corollary}

\section{Linear Time Algorithms for some Chordal Graph Classes}

While the end-vertex problem of MCS is $ \NP $-complete on general graphs, we will now present linear time algorithms for the end-vertex problem on split graphs and unit interval graphs. With the used approaches we were also able to improve some polynomial results for the other searches to linear time algorithms. We begin with the following lemma which we will use repeatedly throughout.

\begin{lemma}\label{lemma:simplicial_linear}
Given a graph $G=(V,E)$ and a vertex $t \in V$. There is a linear time algorithm which decides whether $t$ is simplicial.
\end{lemma}

\begin{proof}
At first, we mark each neighbor of $t$ with a special bit. Now for every neighbor $v$ of $t$ we count, how many neighbors of $v$ are also neighbors of $t$. Because of the special bits this can be done in $\O(|N(v)|)$. If for every neighbor of $t$ this number is equal to the degree of $t$, $t$ is simplicial. Otherwise it is not. The overall running time is linear in the size of $G$, since we visit each edge only a constant number of times.  
\end{proof}

\subsection{Split Graphs}\label{sec:mcs_split}

As split graphs are chordal, we know that Lemma~\ref{mns:lemma1} yields a necessary condition for being an MCS end-vertex. In Figure~\ref{fig:split} we present an example which shows that this condition is not sufficient. However, in the following, we show that a slight strengthening of this condition is enough for a complete characterization.

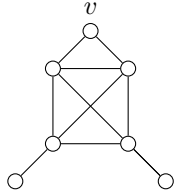
\begin{figure}
	\centering
	\begin{tikzpicture}[vertex/.style={inner sep=2pt,draw,circle}]
	\node[vertex] (1) at (0,0) {};
	\node[vertex] (2) at (1,0) {};
	\node[vertex] (3) at (0,1) {};
	\node[vertex] (4) at (1,1) {};
	\node[vertex, label=above:$ v $] (5) at (0.5,1.5) {};
	\node[vertex] (6) at (-0.5,-0.5) {};
	\node[vertex] (7) at (1.5,-0.5) {};
	
	\draw (1) -- (2);
	\draw (1) -- (3);
	\draw (1) -- (4);
	\draw (2) -- (4);
	\draw (2) -- (3);
	\draw (3) -- (4);
	\draw (1) -- (6);
	\draw (2) -- (7);
	\draw (3) -- (5);
	\draw (4) -- (5);
	\draw (2) -- (7);
	\end{tikzpicture}\caption{Vertex $v$ is an MNS end-vertex, since it fulfills the conditions of Lemma~\ref{mns:lemma1}. However, by enumerating all possible searches one can see that $v$ is not an MCS end-vertex.}\label{fig:split}
\end{figure}

\begin{theorem}\label{mcs:split}
	Let $G=(V,E)$ be a split graph. Then $t \in V$ is the last vertex of some MCS-ordering $\sigma$ of $G$ if and only if
	\begin{enumerate}
		\item The vertex $t$ is simplicial;
		\item The neighborhoods of the vertices with a smaller degree than $t$ are totally ordered by inclusion.
	\end{enumerate}
\end{theorem}

\begin{proof}
	As $ G $ is a split graph, we can assume that its vertex set can be partitioned into a clique $C$ and an independent set $I$ such that $ C $ is a maximal clique in $ G $, i.e., the neighborhood of each vertex in $I$ is a proper subset of the neighborhood of each vertex in $C$. 
	
	Let $ t $ be an end-vertex of MCS on $ G $. As $ G $ is chordal, we see that $ t $ must be simplicial. Now, assume the second condition does not hold. Then there are two vertices $v$ and $w$ with smaller degrees than $t$ whose neighborhoods are incomparable with regard to inclusion.  Without loss of generality  assume that $v$ is taken before $w$ in $\sigma$  in the MCS where  $ t $ is  an end-vertex. 
	
	We first show that $v$ and $w$  have to be elements of $I$. If $t$ is an element of $I$, then this is easy to see, as every vertex in $C$ has a higher degree than the vertices in $I$. If $t $ is an element of $ C$, then $t$ does not have any neighbor in $I$, as it is simplicial. Therefore, $v$ and $w$ cannot be elements of $C$, since otherwise their degree would not be smaller than the degree of $t$. 
	
	Now observe that all neighbors of $v$ are visited before $w$. Indeed, after we have taken the vertex $v$ in $\sigma$, the remainder of $N(v)$ has a larger label than the other vertices in $C \setminus N(v)$. 
	Since there is at least one neighbor of $v$ which is not a neighbor of $w$, the remainder of $C$ will always have larger label than $w$ after we have taken all neighbors of $v$. Hence, all vertices of $ C $ have to be visited before $w$ in $\sigma$. However, this is a contradiction, since from the moment where all the vertices in $ C $ have been visited, the label of $t$ is always greater than the label of $w$ and, thus, $t$ has to be chosen before $w$ in $\sigma$.
	
	Let us now assume that both conditions hold for $t$. We claim that the following ordering is a valid MCS search that has $ t $ as an end-vertex. We choose the neighborhoods of all vertices with lower degree than $t$ in the order of the inclusion ordering. Every time the complete neighborhood of such a vertex $u$ has been visited, we choose $u$ next. If $t $ is an element of $ C $, there are no remaining vertices of $I$ and we take the remaining vertices of $C$ in an arbitrary ordering, where $t$ is the last vertex. If $t $ is an element of $ I $ we take the remaining vertices of $C$ first. Since the neighborhood of $t$ is not greater than the neighborhood of any remaining vertex, we can visit $t$ last. 
\end{proof}

\begin{proposition}\label{prop:mcs_split_linear}
	The MCS end-vertex problem can be decided in linear time if the given input is a split graph.
\end{proposition}

\begin{proof}
	By Lemma~\ref{lemma:simplicial_linear} the simpliciality of $t$ can be checked in linear time. To check the second condition, we first sort the vertices with degrees smaller than $t$ by their degree. This can be done in linear time using counting sort. Let $v_1, \ldots, v_k$ be this order, where $deg(v_1) \geq \ldots \geq deg(v_k)$. Then we create an array $A$ of size $n$, whose elements correspond to the vertices of $G$. At first we mark each element of $A$ which corresponds to a neighbor of $v_1$ with one. For each $v_i$ with $1 < i \leq k$ we check, whether all elements corresponding to a neighbor of $v_i$ are marked with $i-1$. If this is not the case, the neighborhoods of $v_{i-1}$ and $v_i$ are incomparable with regard to inclusion. Otherwise, we mark each neighbor of $v_i$ with $i$. The overall running time of this algorithm is $\O(n + m)$, since we visit each edge a constant number of times.
\end{proof}

Note, that the same approach can be used to improve the results of Charbit et al.~\cite{charbit2014influence} for the end-vertex problems of LBFS and LDFS on split graphs to linear running time. Furthermore, we can use it to improve the complexity of the end-vertex problem of MNS on chordal graphs. To decide this problem, we can check the conditions of Lemma~\ref{mns:lemma1}. As the minimal separators of a chordal graph can be determined in $\O(n + m)$ using MCS~\cite{kumar1998minimal}, the technique of Proposition~\ref{prop:mcs_split_linear} leads to a linear time algorithm.

\begin{corollary}\label{polynomial:lemma1}
The end-vertex problem of LBFS and of LDFS on split graphs can be solved in linear time. Furthermore, the end-vertex problem of MNS on chordal graphs can be solved in linear time.
\end{corollary} 

\subsection{Unit Interval Graphs}

As any MCS (and also every LDFS) is an MNS, we know that in a chordal graph $ G=(V,E) $ a necessary condition for a vertex $ t \in V $ being an MCS (or LDFS) end-vertex is that $ t $ is simplicial and that the minimal separators in its neighborhood can be ordered by inclusion (as seen in Lemma~\ref{mns:lemma1}). In the following we will proceed to show that in unit interval graphs this is also a sufficient condition.

In \cite{fulkerson1965incidence} Fulkerson and Gross showed that a graph $ G $ is an interval graph if and only if the maximal cliques of $ G $ can be linearly ordered such that, for each vertex $ v $, the maximal cliques containing $ v $ occur consecutively. We will such an ordering a \emph{linear order of the maximal cliques}. It is possible to find such a linear order in linear time, as can be seen, for example, in~\cite{corneil2009lbfs} or~\cite{korte1989incremental}.

This property of the maximal cliques can also be expressed through a linear ordering of the vertices. It has been shown by numerous authors (for example by Olariu in~\cite{olariu1991optimal}) that a graph $ G=(V,E) $ is an interval graph if and only if it has an \emph{interval order}, i.e., an ordering $ \sigma $ of its vertices such that for $ u \prec_{\sigma} v \prec_{\sigma} w $, the existence of $ uw \in E $ implies the existence of the edge $ uv \in E $.

For unit interval graphs, Looges and Olariu~\cite{looges1993optimal} proved a similar result:

\begin{lemma}\cite{looges1993optimal}\label{mcs:lemma5}
	A graph $ G=(V,E) $ is a unit interval graph if and only if it has an ordering $ \sigma =(v_1, \ldots ,v_n)$ of its vertices such that for $ u \prec_{\sigma} v \prec_{\sigma} w $, the existence of $ uw \in E $ implies the existence of the edges $ uv \in E $ and $ vw \in E $. Consequently, for two indices $ i<j $ such that $ v_i v_j \in E $ the set of vertices $ \{v_i, v_{i+1} , \ldots v_{j-1},v_j\} $ forms a clique.	
\end{lemma}

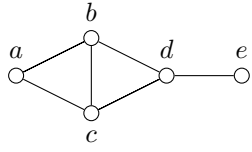
\begin{figure}
	\centering
	\begin{tikzpicture}[vertex/.style={inner sep=2pt,draw,circle}]
		\node[vertex, label=above:$ a $] (1) at (0,0) {};
		\node[vertex, label=above:$ b $] (2) at (1,0.5) {};
		\node[vertex, label=below:$ c $] (3) at (1,-0.5) {};
		\node[vertex, label=above:$ d $] (4) at (2,0) {};
		\node[vertex, label=above:$ e $] (5) at (3,0) {};
		
		\draw[] (1)--(2)--(3)--(1)--(2)--(4)--(3)--(4)--(5);
	\end{tikzpicture}
	\caption{The graph depicted here is a unit interval graph for which the orderings of MNS, LDFS and MCS differ. The order $ (b,c,d,a,e) $ is an MCS order that is not an LDFS order, while $ (b,c,d,e,a) $ is an LDFS order that is not an MCS order. Furthermore, the order $ (d,b,c,e,a) $ is an MNS order that is neither an MCS nor an LDFS order.}\label{mcs:figure6}
\end{figure}

We call such a linear vertex order a \emph{unit interval order}. Both a interval order as well as a unit interval order can be computed in linear time~\cite{corneil2004simple,corneil2009lbfs,korte1989incremental,looges1993optimal}.
Note that even in unit interval graphs MNS, MCS and LDFS can output different search orders. In Figure~\ref{mcs:figure6} we give an example of a unit interval graph for which there is an MCS order that is not an LDFS order and vice versa. Furthermore, this example gives an MNS order that is neither an MCS nor an LDFS order. However, the following theorem shows that the end-vertices of these searches are the same.

\begin{theorem}\label{polynomial:theorem1}
	Let $G=(V,E)$ be a unit interval graph and let $ t $ be a vertex of $ G $. Then the following statements are equivalent:
	\begin{enumerate}[i)]
		\item Vertex $ t $ is simplicial and $ G - N[t] $ is connected.
		\item Vertex $ t $ is the last vertex of some unit interval order.
		\item Vertex $ t $ is an end-vertex of MNS (MCS, LDFS).
	\end{enumerate}
\end{theorem}
\begin{proof}
	To prove that i) implies ii), let $ \sigma=(v_1, \ldots , v_{n}) $ be a unit interval order of $ G $ and let $ t=v_l $. As $ t $ is simplicial, the closed neighborhood of $ t $ forms a maximal clique of $ G $. Thus, all vertices of $ N[t]$ must appear consecutively in $ \sigma $, i.e. $ N[t] =\{v_i, v_{i+1}, \ldots v_{j-1},v_j \} $ for some $ i \leq l \leq j $, due to Lemma~\ref{mcs:lemma5}. Suppose that $ i\neq 1 $ and $ j \neq n $. Then $ G - N[t]$ is disconnected, as there can be no edge between a vertex left of $ v_i $ and a vertex to the right of $ v_j $, again due to Lemma~\ref{mcs:lemma5}. This is a contradiction to the choice of $ t $. Therefore, we can suppose without loss of generality that $ j=n $. Because $ \{v_i, v_{i+1}, \ldots v_{j-1},v_j \} $ is a clique, it is easy to see that $ \sigma'=(v_1, \ldots ,v_{l-1}, v_{l+1}, \ldots v_n, v_l=t) $ is also a unit interval order.
	
	To prove that ii) implies iii), let $ \sigma=(v_1, \ldots , v_n) $ be a unit interval order of $ G $ and let $ t=v_n $. Let $ \sigma'=(w_1, \ldots , w_n) $ be an MNS (MCS, LDFS) order such that the first index $ k $ at which $ v_k \neq w_k $ is rightmost among all such search orders. This implies that at moment $ k-1 $ in the search the vertex $ w_k $ must have a larger label than $ v_k $. In other words, there is a vertex $ v_i $ with $ i < k $ such that $ v_i $ is adjacent to $ w_k $ but not $ v_k $. Since $ v_k \prec_{\sigma} w_k $ this is a contradiction to the fact that $ \sigma $ is a unit interval order.
	
	To show that iii) implies i) assume that $ G - N[t] $ is not connected. Let $ C_1 $ and $ C_2 $ be distinct connected components of $ G - N[t] $. Suppose there is a vertex $ w $ in $ N(t) $ that is adjacent to a vertex $ c_1 \in C_1 $ and a vertex $ c_2 \in C_2 $. Then $ t $, $ w $, $ c_1 $ and $ c_2 $ form an induced claw in $ G $ which is a contradiction to the fact that $ G $ is a unit interval graph~\cite{brandstaedt2000linear}. Therefore, the neighborhood of $ C_1 $ in $ N(t) $, say $ N_1 $, and the neighborhood of $ C_2 $ in $ N(t) $, say $ N_2 $ are disjoint. Both $ N_1 $ and $ N_2 $ form separators of $ G $ and, thus, each of these must contain a minimal separator. Therefore, the minimal separators in $ N(t) $ are not totally ordered by inclusion. As $ t $ is an MNS end-vertex, this is a contradiction to Lemma~\ref{mns:lemma1}.
\end{proof}

Since Condition i) in Theorem~\ref{polynomial:theorem1} can be decided in linear time, we can state the following corollary.

\begin{corollary}
	The MCS and LDFS end-vertex problem can be decided in linear time on unit interval graphs.
\end{corollary}

For DFS there is a simple characterization of the end-vertices of (claw, net)-free graphs using hamiltonian paths. This also holds for unit interval graphs, as they form a subclass of (claw, net)-free graphs~\cite{brandstaedt2000linear}.

\begin{theorem}
	Let $G = (V, E)$ be a (claw, net)-free graph. Then $t \in V$ is the end-vertex of some DFS if and only if $t$ is not a cut vertex.
\end{theorem}
\begin{proof}
	It is clear, that a cut vertex cannot be an end-vertex of a DFS, since it cannot
	be an end-vertex of the generic search~\cite{charbit2014influence}. It remains to show, that every other vertex $t$
	can be an end-vertex. Let $G' = G - t$. $G'$ is still a (claw, net)-free graph. Furthermore, it
	contains a hamiltonian path $P$ since it is connected~\cite{brandstaedt2000linear}. Thus, we can start the DFS in
	$G$ with $P$ and then take $t$ as the last vertex.
\end{proof}

\begin{corollary}
	The end-vertex problem of DFS can be decided in linear time on (claw, net)-free graphs, and, in particular, on unit interval graphs.
\end{corollary}

\subsection{Interval graphs}

In the same vein as for unit interval graphs, we can use a result by Kratsch, Liedloff and Meister~\cite{kratsch2015end} to characterize DFS end-vertices on interval graphs using hamiltonian paths.

\begin{lemma}\cite{kratsch2015end}\label{mcs:lemma8}
	Let $ G=(V,E) $ be a connected graph, and let $ t $ be a vertex of $ G $. Then $ t $ is an end-vertex of DFS if and only if there is a set $ X \subseteq V$ such that $ N[t] \subseteq X $ and $ G[X] $ has a hamiltonian path with endpoint $ t $.
\end{lemma}

On interval graphs this characterization can be simplified to the following.

\begin{theorem}
	Let $ G $ be an interval graph. Then $ t \in V $ is the end-vertex of a DFS if and only if $ G[N(t)] $ contains a hamiltonian path.
\end{theorem}
\begin{proof}
	Suppose $ t \in V $ such that $ G[N(t)] $ contains a hamiltonian path. Then $ G[N[t]] $ must contain a hamiltonian path ending in $ t $ and by Lemma~\ref{mcs:lemma8} $ t $ is a DFS end-vertex of $ G $.
	
	Now assume that $ t $ is a DFS end-vertex. By Lemma~\ref{mcs:lemma8} there exists a set $ X \subseteq V $ such that $ N[t] \subseteq X $ and $ G[X] $ has a hamiltonian path with endpoint $ t $. Let $ X $ be the set of smallest cardinality that fulfills these properties. We claim that $ X=G[N[t]] $. 
	
	Let $ P=(v_1, \ldots , v_l) $ be the hamiltonian path of $ G[X] $ with endpoint $t$ (i.e. $v_l = t$). Suppose there is a vertex $ v_i \in X \setminus N[t] $ and let $ v_i $ be the leftmost such vertex in $ P $. If $ i=1 $, then $P' = (v_2, \ldots v_l)$ is a hamiltonian path in $ X \setminus v_i $ with endpoint $t$, in contradiction to the minimality of $ X $.
	
	Therefore, let $ j < i < k$ such that $ v_j $ is the rightmost vertex of $ N[t] $ to the left of $ v_i $ in $ P $ and $ v_k $ is the leftmost vertex of $ N[t] $ to the right of $ v_i $, i.e. $ v_{j+1}, \ldots , v_{k-1}  \in X \setminus N[t]$. If $ t $ is equal to $ v_k $, then $ v_jv_k \in E $ and $ P' = (v_1, \ldots ,v_j,v_k) $ is a hamiltonian path of $ X \setminus \{v_{j+1}, \ldots , v_{k-1}\} $; this is a contradiction to the minimality of $ X $. Otherwise, $ (t,v_j, \ldots ,v_k,t) $ forms a cycle of length $\geq 4$ in $ G $. As $ v_{j+1}, \ldots , v_{k-1} $ are not adjacent to $ t $, there must be an edge between some non-consecutive vertices in $(v_j, \ldots, v_k)$. As above, this is a contradiction to the minimality of $ X $.
\end{proof}

As the hamiltonian path problem can be solved on interval graphs in linear time, as was shown by Arikati and Rangan~\cite{arikati1990linear}, we can state the following corollary.

\begin{corollary}
	The end-vertex problem of DFS can be decided in linear time on interval graphs.
\end{corollary}

The end-vertex problem for MCS on the class of interval graphs, however, appears to be more complicated than in the case of unit interval graphs, as implication $ iii) \Rightarrow i) $ from Theorem~\ref{polynomial:theorem1} does not necessarily hold here. As a result, not only simplicial vertices that are end-vertices of an interval order are candidates for being an MCS end-vertex. Furthermore, in interval graphs MCS is able to ``jump'' between non-consecutive cliques, making an analysis much harder. Figure~\ref{mcs:figure3} contains an example where an end-vertex $t$ is such that $ G- N[t] $ is disconnected.

\begin{figure}
	\centering
  \begin{minipage}{0.48\textwidth}
	\begin{tikzpicture}[vertex/.style={inner sep=2pt,draw,circle}]
	\node[vertex] (1) at (0,0) {};
	\node[vertex] (2) at (1,0) {};
	\node[vertex] (3) at (2,0) {};
	\node[vertex] (4) at (3,0) {};
	\node[vertex] (5) at (4,0) {};
	\node[vertex] (6) at (5,0) {};
	\node[vertex] (7) at (6,0) {};
	\node[vertex, label=right:$ t $] (8) at (3,1) {};
	
	\draw[] (1) -- (2) -- (3) -- (4) -- (8) --(4) --(5) --(6) -- (7);
  \draw[bend right, draw=none] (4) to (6);
	\end{tikzpicture}
  \end{minipage}
  \begin{minipage}{0.48\textwidth}
	\begin{tikzpicture}[vertex/.style={inner sep=2pt,draw,circle}]
	\node[vertex] (1) at (0,0) {};
	\node[vertex] (2) at (1,0) {};
	\node[vertex] (3) at (2,0) {};
	\node[vertex] (4) at (3,0) {};
	\node[vertex] (5) at (4,0) {};
	\node[vertex] (6) at (5,0) {};
	\node[vertex] (7) at (6,0) {};
	\node[vertex, label=right:$ t $] (8) at (3.5,1) {};
	
	\draw[] (1) -- (2) -- (3) -- (4) -- (8) --(5) --(6) -- (7);
  \draw[] (4) -- (5);
  \draw[bend right] (4) to (6);
	\end{tikzpicture}

  \end{minipage}\caption{Both graphs are interval graphs, but not unit interval graphs, since both contain an induced claw. On the left hand side the MCS end-vertex $ t $ is such that $ G - N[t] $ is disconnected. On the right hand side vertex $t$ is an MNS end-vertex but not an MCS end-vertex.}\label{mcs:figure3}
\end{figure}
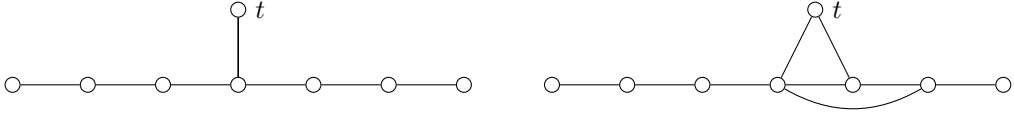

The lemma below gives a relaxed necessary condition which covers the cases similar to the one given in Figure~\ref{mcs:figure3}.

\begin{lemma}\label{mcs:lemma4}
	Let $G=(V,E)$ be an interval graph and let $ C_1, C_2, \ldots , C_k $ be a linear order of the maximal cliques of $ G $. Suppose $t \in V$ satisfies the following conditions:
	\begin{enumerate}
		\item $ t $ is simplicial.
		\item If $ C_i $ is the unique clique containing $ t $, then $i=1$ or $i=k$ or
		\begin{equation*}
		C_{i-1} \cap C_i \subseteq C_i \cap C_{i+1}
		\end{equation*}
		and
		\begin{equation*}
		|C_{i} \cap C_{i+1}| \leq |C_j \cap C_{j+1}| \text{ for all } j > i,
		\end{equation*}
		or the same holds for the reverse order $ C_k, C_{k-1}, \ldots , C_1 $.
	\end{enumerate}
	Then $ t $ is the end-vertex for some MCS on $ G $.
\end{lemma}

\begin{proof}
	As $ G $ is an interval graph, we can assume that there is a linear order $ C_1, C_2, \ldots , C_k $ of the maximal cliques of $ G $, in the sense that for any vertex $ v \in V $ all cliques containing $ v $ are consecutive. Also, it is easy to see that any vertex is simplicial if and only if it is contained in exactly one of these maximal cliques.
	
	We assume that $ t $ fulfills both properties. We now construct an MCS search ordering which starts in a simplicial vertex of $ C_1 $ and ends in the vertex $ t $. Without loss of generality, we  assume that the second property holds for $ C_1, C_2, \ldots , C_k $, as otherwise we can just begin our search at $ C_k $ and use the same arguments.
	
	We proceed by visiting the maximal cliques of $ G $ consecutively -- choosing the vertices of each clique in an arbitrary order -- until we have completely visited $ C_{i-1} $. As $ C_{i-1} \cap C_i \subseteq C_i \cap C_{i+1} $, we can visit a vertex of $ C_{i+1} \setminus C_{i} $ next, ignoring the simplicial vertices of $ C_i $. Due to the fact that $ |C_{i} \cap C_{i+1}| \leq |C_j \cap C_{j+1}| $ for all $ j > i $ we can then visit all vertices of $ C_{i+1} $ and the vertices of all other maximal cliques apart from $ C_i $. Finally, we visit the remaining simplicial vertices of $ C_i $, choosing $ t $ last.
\end{proof}

Note that the sufficient condition given by Lemma~\ref{mcs:lemma4} may not be necessary in general. For example, Figure~\ref{mcs:figure4} shows an example of a graph where the second condition from Lemma \ref{mcs:lemma4} is not true for vertex $t$, however it is still an MCS end-vertex. Nonetheless, we believe that it is still possible to give a characterization of MCS end-vertices for the family of interval graphs which may be checked efficiently.

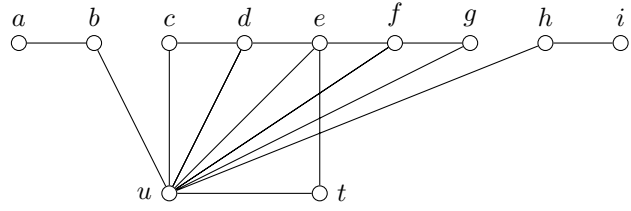
\begin{figure}
	\centering
	\begin{tikzpicture}[vertex/.style={inner sep=2pt,draw,circle}]
	\node[vertex, label=above:$ a $] (1) at (0,2) {};
	\node[vertex, label=above:$ b $] (2) at (1,2) {};
	\node[vertex, label=above:$ c $] (3) at (2,2) {};
	\node[vertex, label=above:$ d $] (4) at (3,2) {};
	\node[vertex, label=above:$ e $] (5) at (4,2) {};
	\node[vertex, label=above:$ f $] (6) at (5,2) {};
	\node[vertex, label=above:$ g $] (7) at (6,2) {};
	\node[vertex, label=above:$ h $] (8) at (7,2) {};
	\node[vertex, label=above:$ i $] (9) at (8,2) {};
	
	\node[vertex, label=left:$ u $] (u) at (2,0) {};
	
	\node[vertex, label=right:$ t $] (t) at (4,0) {};
	
	\draw[] (1) -- (2) -- (u) -- (3) -- (4) -- (u) --(4) -- (5)-- (t) -- (u) -- (5)  -- (6) -- (u) -- (6) -- (7) -- (u) --(8) --(9);
	\end{tikzpicture}\caption{An example of an interval graph $ G $ and a vertex $ t \in V(G) $, where $ t $ is an MCS end-vertex of the search $ (a,b,u,h,i,c,d,e,f,g,t) $. However, it is easy to see that for the linear order of the maximal cliques of the form $ C_1=\{a,b\} $, $ C_2=\{b,u\} $, $ C_3 =\{c,d,u\} $, $ C_5 =\{ d,e,u\} $, $ C_6=\{e,t,u\} $, $ C_7 =\{e,f,u\} $, $ C_8=\{f,g,u\} $, $ C_9=\{u,h\} $ and  $C_{10}=\{h,i\} $ the second condition of Lemma~\ref{mcs:lemma4} is not fulfilled. In fact, this condition is false for any linear order of the maximal cliques.}\label{mcs:figure4}
\end{figure}

\begin{conjecture}
	The MCS end-vertex problem can be decided in polynomial time if the given input is an interval graph.
\end{conjecture}


\section{Conclusion} 

We have shown that the end-vertex problem is \NP-complete for MNS on weakly chordal graphs and for MCS on general graphs. Moreover, we have given linear time algorithms to compute end-vertices for LDFS and MCS on unit interval graphs as well as for DFS on interval graphs and for MCS on split graphs. Using the same techniques, we were able to improve the analyses of running times of various previous results from polynomial time to linear time. A complete list of the achieved results can be found in Table~\ref{tab:results}. However, many open questions still remain.

For all the searches investigated here, apart from DFS, the complexity of the end-vertex problem on chordal graphs is still open. This is especially interesting, as nearly all of these problems are already hard on weakly chordal graphs. Even on interval graphs the complexity for some end-vertex problems is open.

Besides the complexity results for the end-vertex problem on graphs with bounded chordality, there are further results on bipartite graphs. Charbit, Habib, Mamcarz~\cite{charbit2014influence} showed, that it is \NP-complete for BFS. Gorzny and Huang~\cite{gorzny2017end} showed the same result for LBFS. Furthermore, they present a polynomial time algorithm for the LBFS end-vertex problem on AT-free bipartite graphs. It is an open question, whether these results can be extended to the end-vertex problems of MCS and MNS. 

As mentioned in the introduction, the recognition of search trees of BFS and DFS is easy~\cite{hagerup1985,korach1989,manber1990}, although the corresponding end-vertex problems are hard. In~\cite{beisegel2018recognizing} the authors show that LDFS-trees can be recognized in polynomial time, while the recognition of LBFS-trees is \NP-complete on the class of weakly chordal graphs. Furthermore, they showed that recognition of both MNS and MCS-trees is \NP-complete. However, it remains an interesting question, whether there exists a search and a graph class, such that recognition of the search trees on that class is hard, whereas the end-vertex problem is polynomial.

	\acknowledgements{The work of this paper was done in the framework of a bilateral project between Brandenburg University of Technology and University of Primorska, financed by German Academic Exchange Service and the Slovenian Research Agency (BI-DE/17-19-18).
 Matja\v{z} Krnc gratefully acknowledges the European Commission for funding the InnoRenew CoE project (Grant Agreement \#739574) under the Horizon2020 Widespread-Teaming program and the Republic of Slovenia (Investment funding of the Republic of Slovenia and the European Union of the European Regional Development Fund).
 Nevena Piva\v c was funded by the Young Researcher Grant, financed by the Slovenian Research Agency (ARRS).
 
The authors would like to thank the anonymous referees for their many helpful comments, especially for suggesting a simplification of Theorem~\ref{polynomial:theorem1}.}
  
	\bibliographystyle{plain}
	\bibliography{end-vertex}

\end{document}